\documentclass[letterpaper, 10 pt, conference]{ieeeconf}
\IEEEoverridecommandlockouts
\usepackage[space]{cite}

\usepackage[tbtags]{amsmath}
\usepackage{amsfonts}
\usepackage{amssymb}
\usepackage{amsthm}
\usepackage{mathrsfs}
\usepackage{xspace}
\usepackage{xparse}
\usepackage{mathtools}
\usepackage{comment}
\usepackage{xcolor}
\usepackage{bbm}

\newtheorem{thm}{Theorem}

\newtheorem{defn}{Definition}

\newtheorem{assum}{Assumption}
\newtheorem{lem}{Lemma}

\newcommand{\mcl}{\mathcal}
\newcommand{\lp}{\left(}
\newcommand{\rp}{\right)}
\newcommand{\mbb}{\mathbb}

\newcommand{\ep}{\tau}

\newcommand{\iter}{n}
\newcommand{\Iter}{N}

\begin{document}

\title{
Active Learning of Dynamics Using Prior Domain Knowledge in the Sampling Process
}

\author{
Kevin~S.~Miller,~\IEEEmembership{Member,~IEEE,}
Adam~J.~Thorpe,~\IEEEmembership{Member,~IEEE,}
Ufuk~Topcu,~\IEEEmembership{Senior Member,~IEEE}%
\thanks{%
    This material is based upon work supported by the National Science Foundation under NSF Grant Number 1652113 and 1836900.  Any opinions, findings, and conclusions or recommendations expressed in this material are those of the authors and do not necessarily reflect the views of the National Science Foundation.
    K. S. Miller was supported by a Peter J. O'Donnell Jr. Postdoctoral Fellowship and NSF IFML Grant Number 2019844.
}
\thanks{K. Miller, A. Thorpe, and U. Topcu are with the Oden Institute for Computational Engineering and Sciences at the University of Texas at Austin, Austin, TX, USA.
}
\thanks{Email: {\tt{\{ksmiller,utopcu\}@utexas.edu}, \newline \tt{adam.thorpe@austin.utexas.edu}}.}
\thanks{K. S. Miller and A. J. Thorpe contributed equally to this work.}
}

\maketitle

\begin{abstract}
    We present an active learning algorithm for learning dynamics that leverages side information by explicitly incorporating prior domain knowledge into the sampling process. Our proposed algorithm guides the exploration toward regions that demonstrate high empirical discrepancy between the observed data and an imperfect prior model of the dynamics derived from side information. Through numerical experiments, we demonstrate that this strategy explores regions of high discrepancy and accelerates learning while simultaneously reducing model uncertainty. We rigorously prove that our active learning algorithm yields a consistent estimate of the underlying dynamics by providing an explicit rate of convergence for the maximum predictive variance. We demonstrate the efficacy of our approach on an under-actuated pendulum system and on the half-cheetah MuJoCo environment. 
\end{abstract}

\section{Introduction}

Model-based and data-driven methods typically represent two alternative approaches to stochastic optimal control. While purely data-driven control typically neglects prior domain knowledge (side information) to reduce bias, incorporating such knowledge into data-driven control can yield more accurate learned models of dynamical systems. 
For example, prior domain knowledge in the form of an imperfect physics-based model can be combined with data-driven modeling to more accurately approximate the true system dynamics. 
Despite the focus on side information in the learned model, there are foreseeable benefits to also leveraging side information to select informative data, and the question of how to select data according to available side information remains an open challenge.

We present an active dynamics learning method that utilizes side information to select sample data in regions where the observed discrepancy between prior domain knowledge and the observed data is highest.
Our approach is based on the upper confidence bound (UCB) algorithm \cite{auer_using_2002, auer_finite-time_2002, Dani2008StochasticLO}, known for its ability to balance exploration and exploitation in multi-armed bandit and Bayesian optimization settings. 
We specifically consider the Gaussian process (GP) setting \cite{rw}, known as GP-UCB \cite{6138914}, which offers a principled approach to characterizing uncertainty---a key component utilized by the UCB algorithm.  
Our key innovation lies in incorporating side information \textit{during the sampling phase}.
We actively sample control inputs along trajectories that favor exploration in regions that demonstrate a higher discrepancy between our observed data and an imperfect prior model of the dynamics.
This emphasizes sampling in regions of the state-action space where prior knowledge does not align with the data-driven estimate while avoiding redundant sampling in regions where our data-driven model aligns with the observed dynamics. 

Our approach is enabled by two key elements: 1) actively learning dynamics in an episodic setting, and 2) incorporating side information.

Active sampling has been explored in the context of learning dynamical systems \cite{wagenmaker2020activelearning, simchowitz2018learningwithoutmixing, simchowitz2019learninglinear, burdick2023active}. 
In the case of linear dynamics, active learning approaches can provide optimal or near-optimal sample complexity results \cite{wagenmaker2020activelearning, simchowitz2018learningwithoutmixing, simchowitz2019learninglinear}.
However, these results often impose unrealistic assumptions--such as the ability to sample any state-action pair without regards to dynamic constraints.
In the general case, sample complexity guarantees are more elusive. 
Optimistic planning approaches, such as OpAx \cite{sukhija2023optimistic} and H-UCRL \cite{NEURIPS2020_a36b598a}
actively sample data in an episodic setting by selecting policies in an open-loop fashion to maximize information gain via an optimistic planner. 
We likewise consider an episodic setting, 
but we consider using prior domain knowledge to determine our exploration policy. 
Notably, GP-UCB has been used in a different context for robot kinematic calibration \cite{burdick2023active}, in which the authors use active sampling to learn the correction to a prior, parametric model of the robot's kinematics. 
While they allow for arbitrary sampling of states and control inputs throughout a single sampling process, we restrict to the selection of control input sequences along trajectories in an episodic fashion. 
That is, we must account for planning action sequences under uncertain dynamics in an MPC-like framework.

Active sampling methods typically do not incorporate prior domain knowledge into the model or the sampling process.
In particular, except for \cite{burdick2023active}, the previously mentioned active learning methods (i.e., \cite{wagenmaker2020activelearning, simchowitz2018learningwithoutmixing, simchowitz2019learninglinear, sukhija2023optimistic}) do not utilize side information in the data-driven model of the system dynamics. 
The use of side information has been investigated for two-armed bandit problems in \cite{5714284, 1406128}, where the sampler has access to information of the reward. 
However, these results do not translate easily to the setting of dynamics learning. To our knowledge, 
our proposed active dynamics learning method is the first to leverage side information in the sampling process.

Our main contribution is an active dynamics learning algorithm based on GP-UCB that incorporates prior domain knowledge into both the sampling process and the learned model. 
Specifically, our approach 
incorporates a discrepancy term in the UCB sampling function that empirically models the difference between the data-driven portion of our model and the prior domain knowledge.
Unlike existing approaches, we 
exploit prior domain knowledge via this discrepancy term to focus sampling in regions where the prior model is most misaligned with the true dynamics.
In addition, we provide a proof that our approach yields a consistent estimator of the dynamics
as more episodes are considered. 
It is often difficult to guarantee consistency for active sampling algorithms due to the non-i.i.d.\ nature of the sampling. 
We prove that as long as the trajectory generated during each episode explores state-action pairs that possess greater than average predictive variance under the GP model, our estimate converges to the true dynamics as the number of episodes increases. 
We demonstrate our approach on a simple pendulum system, and compare our approach against OpAx \cite{sukhija2023optimistic} and 
greedy variance-based sampling, which represents a pure exploration strategy.
We then demonstrate that we can use our learned dynamics model for control on a high-dimensional half-cheetah MuJoCo environment.

The rest of the paper is outlined as follows. In Section \ref{sec: prelim and problem statement}, we formally state the problem setting and provide preliminary background information to inform the discussion and introduction of our sampling method. 
We present our active sampling method in Section \ref{sec: active sampling method} and provide the consistency argument in Subsection \ref{subsec : consistency guarantee}. 
In Section \ref{sec: numerical results}, we demonstrate our active sampling approach.

\section{Preliminaries \& Problem Statement} \label{sec: prelim and problem statement}

\subsection{Problem Statement} \label{subsec: problem statement}

Consider the following discrete-time dynamical system,
\begin{equation}
    \label{eqn: system dynamics}
    x_{t+1} = f(x_{t}, u_{t}) + w_{t},
\end{equation}
where $x_{t} \in \mathcal{X}$ is the state of the system at time $t$, $u_{t} \in \mathcal{U}$ is the control action applied to the system at time $t$, and $w_{t} \sim \mathcal{N}(0, \sigma^{2} I)$ is an independent Gaussian noise term. 

Given a cost function $c : \mathcal{X} \times \mathcal{U} \to \mathbb{R}$, we formulate the following stochastic optimal control problem, where the goal is to select a sequence of control inputs $u_{0}, \ldots, u_{N}$ to minimize the expected cumulative cost, 
\begin{subequations}
\label{eqn: stochastic optimal control problem}
\begin{align}
    \min_{u_{0}, \ldots, u_{N}} \quad & \mathbb{E} \biggl[ \sum_{t=1}^{N} c(x_{t}, u_{t}) \biggr] \\
    \label{eqn: stochastic optimal control problem dynamics}
    \text{s.t.} \quad & x_{t+1} = f(x_{t}, u_{t}) + w_{t}
\end{align}
\end{subequations}
We presume that the system dynamics $f$ in \eqref{eqn: system dynamics} are unknown, meaning the control problem \eqref{eqn: stochastic optimal control problem} is intractable. However, we presume access to an \emph{imperfect} model of the dynamics $p_0 : \mathcal{X} \times \mathcal{U} \to \mathcal{X}$. 
Such side information may be available, for instance, if we have a coarse approximation of system parameters, a first-order estimate of the dynamics, or access to a low-fidelity model through a virtual simulation.

We consider the problem of sequentially computing a data-driven estimate $\mu_{n}$ of the system dynamics \eqref{eqn: system dynamics} by actively selecting the dataset $\mathcal{D}_{n}$ consisting of $n \in \mathbb{N}$ observed state transitions, 
\begin{equation}
    \label{eqn: dataset}
    \mathcal{D}_{n} = \lbrace (x^{i}, u^{i}, y^{i}) \rbrace_{i=1}^{n}.
\end{equation}
where $x$ and $u$ are in $\mathcal{X}$ and $\mathcal{U}$, respectively, and $y = f(x, u) + w$ with $w \sim \mcl N(0, \sigma^2 I)$. 

We study an episodic setting, where in each episode $\tau = 1, 2, \ldots, T$, we compute an exploratory policy $\pi_\tau : \mcl X \rightarrow \mcl U$ that we employ to collect new data.
As we obtain new observations, we update the dataset, $\mathcal{D}_{n+1} = \mathcal{D}_{n} \cup \lbrace (x', u', y') \rbrace$ and subsequently update
the data-driven estimate of the dynamics.

To that end, our goal is to develop an active sampling algorithm that defines an exploratory policy $\pi_\ep$ in each episode $\ep=1, 2, \ldots$ to leverage (i) prior domain knowledge and (ii) information from prior episodes to efficiently learn a predictive model 
of the dynamics $f(x,u)$ in \eqref{eqn: system dynamics}. 
Furthermore, we seek a \emph{consistent} estimator, meaning
$\mu_\ep(x, u) \rightarrow f(x,u)$ as $\ep \rightarrow \infty$.

Our key insight is that a non-zero mean GP can be used to model the residual between our observed data and a prior model derived from side information. Then, we modify the UCB algorithm to maximize the discrepancy between our data-driven model and the prior, and 
periodically accumulate observed data into our prior model at the end of each episode.
This leads the active sampling process to sample more in areas where the prior model is empirically incorrect, and less in areas where our prior is closely aligned with the data.

\subsection{Gaussian Processes} \label{sec: gaussian process}

We use a Gaussian process \cite{rw} to estimate the dynamics.
For notational simplicity, we denote $\mathcal{Z} = \mathcal{X} \times \mathcal{U}$ as the state-action space and $z_{t} = (x_{t}, u_{t})$. 
A Gaussian process $f(z) \sim \mathcal{GP}(m(z), k(z, z'))$ is completely specified by a mean function $m(z)$ and a positive definite covariance function $k(z, z')$ \cite{rw},
\begin{align}
    \label{eqn: gaussian process mean}
    m(z) &= \mathbb{E}[f(z)] \\
    \label{eqn: gaussian process covariance}
    k(z, z') &= \mathbb{E}[(f(z) - m(z))(f(z') - m(z'))]
\end{align}
In the following, we presume that the covariance function $k$ is given by the squared exponential, or RBF function, $k(z, z') = \exp(- \gamma \lVert z - z' \rVert^{2})$, where $\gamma > 0$, and that the mean function $m(z)$ is non-zero \cite[\S~2.7]{rw}. 
The non-zero mean is key to our approach since we use this term to capture prior knowledge of the system dynamics. 
Critically, we will periodically update the prior to reflect the new data gathered during each episode.

Given a dataset $\mathcal{D}_{n}$ as in \eqref{eqn: dataset} consisting of $n$ data points, the predictive mean $\mu_{n}$ and variance $\sigma_{n}^{2}$ can be evaluated at a given test point $z^{*} \in \mathcal{Z}$ as, 
\begin{align}
    \label{eqn: predictive mean}
    \mu_{n}(z^{*}) &= (\boldsymbol{y} - \boldsymbol{m})^{\top} (G_{n} + \sigma^{2} I)^{-1} \boldsymbol{k}_{z^{*}} + m(z^{*}) \\
    \label{eqn: predictive variance}
    \sigma_{n}^{2}(z^{*}) &= k(z^{*}, z^{*}) - \boldsymbol{k}_{z^{*}}^{\top} (G_{n} + \sigma^{2} I)^{-1} \boldsymbol{k}_{z^{*}}
\end{align}
where $\boldsymbol{y}$ and $\boldsymbol{m}$ are vectors, with the $i^{\rm th}$ elements given by $\boldsymbol{y}_{i} = y^{i}$ and $\boldsymbol{m}_{i} = m(x^{i}, u^{i})$, $G_{n} = (g_{ij})$ is an $n \times n$ matrix with elements $g_{ij} = k(z^{i}, z^{j})$, $\boldsymbol{k}_{z^{*}}$ is a vector where the $i^{\rm th}$ element is given by $k(z^{i}, z^{*})$.
As we collect new observations and augment the dataset $\mathcal{D}_{n}$, the predictive mean $\mu_{n}$ and variance $\sigma_{n}^{2}$ are recomputed using the new dataset. 
In practice, these equations can be solved efficiently using Cholesky factorization, see \cite[Algorithm~2.1]{rw}, and the Cholesky factors can be updated via rank-1 updates as new data becomes available. 

\subsection{Active Sampling in the GP Setting} \label{sec: active-background}

We are concerned with fitting a GP model to a minimal dataset, i.e. sampling a limited number of data points that provide significant information about the true dynamics of the system, $f$. 
We utilize an adaptation of the GP-UCB algorithm \cite{6138914} to guide the selection of actions along trajectories to regions of greatest mismatch (discrepancy) between the prior model and the true, underlying dynamics. The GP-UCB algorithm\cite{6138914} for maximizing a function $g(z)$  for $z \in \mcl Z$ uses the predictive mean $\mu_n$ and variance $\sigma_n$ of a GP to decide the next sample point, and chooses points based on the following \emph{acquisition function}, 
\begin{equation}
    \label{eqn: orig acquisition function}
    A(z) = \mu_{n}(z) + \beta_{n}^{1/2} \sigma_{n}(z)
\end{equation}
where $\mu_{n}$ is as in \eqref{eqn: predictive mean}, $\sigma_{n}$ is the square root of the predictive variance in \eqref{eqn: predictive variance}, and $\beta \in \mathbb{R}_{+}$ is a positive real constant called the decay schedule.

The next sample point $z_{n+1}$ is then chosen to maximize the acquisition function as
\begin{equation} \label{eqn: ucb selection}
    z_{n+1} = \arg \max_{z \in \mathcal{Z}} A_{n}(z).
\end{equation}
Intuitively, the acquisition function $A$ in \eqref{eqn: orig acquisition function} provides a tradeoff between sampling in areas where the predictive mean is large, and areas of high variance. This means that in practice, under correct choice of the decay schedule $\beta_{n}$, the UCB algorithm will trade off between ``exploration'' in areas of high uncertainty and ``exploitation'' by selecting points near the max of the predictive mean.

\section{Active Sampling Using Side Information} \label{sec: active sampling method}

Our key insight is to define an exploration policy $\pi_\ep : \mcl X \rightarrow \mcl U$ in episode $\ep$ using \eqref{eqn: ucb selection} to focus sampling on $\Iter$ state-action pairs, $\{(x^{\ep,\iter}, u^{\ep, \iter})\}_{\iter=1}^\Iter$, in regions where the previous episode's learned model, $\mu_{\ep-1,\Iter}$, is maximally different from the true system dynamics, $f$. 
As described in Section \ref{sec: gaussian process}, we iteratively update our GP model as
\begin{align} \label{eqn : episodic gp model}
    \mu_{\ep, \iter}(z) &= (\boldsymbol{y} - \boldsymbol{m}_\ep)^{\top} (G_{\ep,\iter} + \sigma^{2} I)^{-1} \boldsymbol{k}_{z} + m_\ep(z) \\
    \sigma_{\ep, \iter}^{2}(z) &= k(z, z) - \boldsymbol{k}_{z}^{\top} (G_{\ep, \iter} + \sigma^{2} I)^{-1} \boldsymbol{k}_{z}, \label{eqn : episodic gp model var}
\end{align}
where $G_{\ep,\iter} \in \mbb R^{(\ep-1)\Iter +\iter \times (\ep-1)\Iter +\iter}$ is the kernel Gram matrix of all data points observed up to iteration $\iter$ of episode $\ep$, and $m_\ep : \mcl Z \rightarrow \mcl X$ is an \textit{episode-dependent} prior term that we define to be
\begin{equation} \label{eqn : episode-dependent prior}
    m_\ep(z) \coloneqq \mu_{\ep-1, \Iter}(z),
\end{equation}
which is the previous episode's learned model.
At $\ep = 1$, we define $m_1(z) = p_0(z)$, which is determined by the side information or prior domain knowledge available before sampling. 
Furthermore, we define $\mu_{\ep, 0} \equiv \mu_{\ep-1,\Iter}$ and $\sigma_{\ep,0} \equiv \sigma_{\ep-1,\Iter}$ to accumulate observed data from the previous episode.

The main idea of the sampling procedure is that the exploration policy $\pi_\ep$ is computed to maximize an adapted GP-UCB acquisition function over action sequences along the trajectory predicted by the current model of the dynamics. We consider an MPC-like procedure to iteratively update and replan the action sequences at each iteration during the episode. For simplicity of notation, we consider the time horizon to be the same as the episode length, $\Iter$. 

We define the following \textit{discrepancy-based} acquisition function,
\begin{align} \label{eqn : discrepancy-based acquisition function}
    A_{\ep, \iter}(u; x^{\ep, \iter}) &= \big|\mu_{\ep, \iter}(x^{\ep, \iter}, u) - m_\ep(x^{\ep, \iter}, u)\big| \nonumber \\
    &\qquad + \beta_{\ep, \iter}^{1/2} \sigma_{\ep, \iter}(x^{\ep, \iter}, u) + s_\ep(x^{\ep, \iter}, u).
\end{align}
The first term of \eqref{eqn : discrepancy-based acquisition function} models the \textit{discrepancy} between our current model $\mu_{\ep, \iter}$ and the prior $m_{\ep}$ fixed at the beginning of the current episode.
Maximizing this term encourages sampling actions where the current episode's prior appears to be most incorrect. As in \eqref{eqn: orig acquisition function}, the middle term of \eqref{eqn : discrepancy-based acquisition function} allows for the tradeoff between exploration of actions that possess high uncertainty (variance) and exploitation of actions that maximize the discrepancy. 
The last term $s_{\ep}$ of \eqref{eqn : discrepancy-based acquisition function} is an additional term to allow for further side information to be accounted for in the exploration policy. For example, $s_\ep(z)$ could be defined to ensure the system avoids exploring states or actions known to be unsafe, e.g. as,
\begin{equation} 
    s_\ep(z) = \begin{cases}
        0 & z \in \mcl S \\
        -\infty & z \not\in \mcl S \\
    \end{cases},
\end{equation}
where $\mcl S \subset \mcl Z$ is an episode-dependent safe set.

At each iteration $\iter$ during an episode $\ep$, we solve the following optimization problem,
\begin{subequations}
\label{eqn : mpc acquisition function setup}
\begin{align} 
    \max_{u_0, \ldots, u_{\Iter-1} \subset \mcl U} \quad &\sum_{t=0}^{\Iter -1}A_{\ep, \iter}(u_t; x_{t}) \\
    \text{s.t.} \quad &  x_{t+1} = \mu_{\ep, \iter}(x_{t}, u_{t}) \\
    & x_{0} = x^{\ep,\iter}
\end{align}
\end{subequations}
Intuitively, we plan a sequence of actions $u_0, \ldots, u_{\Iter-1}$ from the current state $x^{\ep, \iter}$ that, given the current estimate of the dynamics $\mu_{\ep,\iter}$, will maximize the sum of the discrepancy-based acquisition function in \eqref{eqn : discrepancy-based acquisition function} along the predicted trajectory.
After solving \eqref{eqn : mpc acquisition function setup}, we execute the first control action in the planned sequence, accumulate a new data point in our dataset, and update the dynamics estimate $\mu_{\ep,\iter}$.

At the first iteration $n=0$ during each episode $\ep$, we set $x^{\ep, 0}$ to be the state in $\mathcal{X}$ with the maximum predictive variance $\sigma^2_{\ep,\iter}$.
Note that by definition of $m_\ep \equiv \mu_{\ep-1, \Iter} \equiv \mu_{\ep, 0}$, the discrepancy term in the first iteration is 
$\big|\mu_{\ep, 0}(x^{\ep, \iter}, u) - m_\ep(x^{\ep, \iter}, u)\big| = 0$, 
and so at the first step, the problem in \eqref{eqn : mpc acquisition function setup} reduces to selecting action sequences that maximize the sum of the variances along the trajectory (plus whatever other side information is included in the term $s_\ep$).

At subsequent iterations, the discrepancy is no longer uniformly $0$ throughout the state-action space, and the solutions of action sequences will incorporate how the observed dynamics (as represented by the updated models $\mu_{\ep, \iter}$) differ from the previous episode's learned model $\mu_{\ep-1, \Iter}$. 
In this way, our method prioritizes the selection of actions that are likely to explore in regions of state-action space where our current data model $m_\ep$ is misaligned with the true dynamics, $f$. 

At each iteration $\iter$ in episode $\ep$, we solve \eqref{eqn : mpc acquisition function setup} to compute the exploration policy.
In practice, solving \eqref{eqn : mpc acquisition function setup} can pose a challenge due to the presence of the square root of the variance $\sigma_{\ep, \iter}$ in the objective. Thus, we can use a sample-based MPC method to compute a sequence of control actions, for instance via Cross Entropy Maximization (CEM) \cite{pmlr-v155-pinneri21a}.

\subsection{Consistency Guarantee} \label{subsec : consistency guarantee}

Active sampling procedures should possess guarantees that they lead to consistent estimators of the underlying dynamics. It is usually the case that consistency arguments are made under the assumption that the observed data is given i.i.d., which is not the case in most active learning settings. We establish consistency of our GP model $\mu_{\ep, \iter}$ within our active sampling framework by leveraging the useful property that GPs are all-time calibrated statistical models \cite{sukhija2023optimistic, NEURIPS2020_a36b598a}. Our key insight is to leverage the decrease in predictive variance values of our GP model at key points along observed trajectories to control the overall decrease in predictive variance values over the state-action space $\mcl Z$.

We begin by stating important definitions and assumptions that will allow us to conclude consistency of the GP mean $\mu_{\ep, \Iter} \rightarrow f$ from an argument on the convergence of the predictive variance values $\sigma^2_{\ep, \Iter} \rightarrow 0$. We make the following assumption and give a useful definition:
\begin{assum} \label{assum: bounded rkhs norm}
    We assume that the coordinate-wise functions $[f(\cdot)]_\ell = f_\ell : \mcl Z \rightarrow \mbb R$ lie within a RKHS with kernel $k$ and have bounded norm $B$. That is, $f \in \mcl H_{k,B}^d = \{ f: \|f_\ell\|_k \leq B, \ell = 1, \ldots, d\}$.
\end{assum}
\begin{defn}[All-time calibrated statistical model of $f$, {\cite{rothfuss2023hallucinated}}] 
    \label{def: all-time calibrated}
    Let $z = (x,u)$ and $\mcl Z \coloneqq \mcl X \times \mcl U$. An all-time calibrated statistical model for the function $f$ is a sequence $\{\mu_j, \sigma_j, \beta_j(\delta)\}_{j \geq 0}$ such that for all $z \in \mcl Z,$ $\ell \in \{1, \ldots, d\}$, and $j \in \mbb N$
    \begin{equation}
        \left| [\mu_{j}(z)]_\ell - [f(z)]_\ell \right| \leq \beta_j(\delta) [\sigma_{j}(z)]_\ell 
    \end{equation}
    with probability greater than or equal to $1 - \delta$.
    Here we denote the $\ell^{th}$ element of a vector $\mathbf{v} \in \mbb R^d$ as $[\mathbf{v}]_\ell$. The scalar function, $\beta_n(\delta) \in \mbb R_{+}$ quantifies the width of the $1-\delta$ confidence intervals. We assume wlog that $\beta_j$ monotonically increases with $n$, and that $[\sigma_j(z)]_\ell \leq \sigma_{max}$ for all $z \in \mcl Z, j \geq 0$, and $\ell \in \{1, \ldots, d\}$.
\end{defn}
Then, assuming that the true dynamics $f$ we wish to model is a bounded function (i.e., Assumption \ref{assum: bounded rkhs norm}) within an RKHS of vector-valued functions, then we can conclude that the GP model is a well-calibrated statistical model for $f$:
\begin{lem}[Well-calibrated confidence intervals for RKHS, {\cite{rothfuss2023hallucinated}}] \label{lemma : well-calibrated rkhs}
    Let $f \in \mcl H^d_{k,B}$ and suppose that $\mu_j$ and $\sigma_j$ are the posterior mean and variance of a GP with kernel $k$. There exists $\beta_j(\delta)$, for which the tuple $(\mu_j, \sigma_j, \beta_j(\delta))$ 
    is an all-time-calibrated statistical model (Definition \ref{def: all-time calibrated}) of $f$ .
\end{lem}

Note that while these definitions of the predictive variance are considered to be vector-valued, our GP model simply considers scalar variance values (i.e., the case that all entries of the vector-valued variance are equivalent). Thus, while our result could be stated in terms of bounding the norm of a vector-valued variance,  $\|\boldsymbol{\sigma}^{\mathbf{2}}_{\ep, \Iter}\|_2^2$, we will just consider the convergence of the scalar value $\sigma^2_{\ep, \Iter}$. 

To summarize, to show consistency of the underlying GP model $\mu_{\ep, \Iter}(z) \rightarrow f(z)$ as $\ep \rightarrow \infty$, one simply needs to show that the variance of the GP model $\sigma^2_{\ep, \Iter}$ vanishes as more episodes occur. We consider this guarantee under the following additional assumptions
\begin{assum} \label{assumption : finite Z}
    The state-action space $\mcl Z$ is finite; that is, $|\mcl Z| = N_z < \infty$.
\end{assum} 
\begin{assum}
    The side information term $s_\ep $ is already captured in the definition of the domain of possible state-action pairs, $\mcl Z$ (i.e., $s_\ep \equiv 0$ for our purposes). 
\end{assum}
\begin{assum} \label{assumption : greater than average}
    In each episode $\ep$, there exists at least one index $\iter_\ast \in \{1, \ldots, \Iter\}$ at which the value $z^{\ep, \iter_\ast}$ in the sampled trajectory $\{z^{\ep, \iter} = (x^{\ep, \iter}, u^{\ep, \iter})\}_{\iter=1}^\Iter$ satisfies
    \begin{equation} \label{eqn : assumption bound}
        \sigma^2_{\ep, \iter_\ast}(z^{\ep, \iter_\ast}) \geq \frac{1}{N_z} \sum_{z \in \mcl Z} \sigma^2_{\ep, \iter_\ast}(z),
    \end{equation}
    where $N_z$ is the cardinality of the state-action space $\mcl Z$.
\end{assum}
The final assumption (Assumption \ref{assumption : greater than average}) is a technical assumption that simply asserts that the trajectory in each episode will contain at least one state-action pair $z = (x,u)$ that has a relatively large variance value. We suggest that this is a rather mild condition which is reasonable to assume in various settings, such as in our case of discrepancy-based GP-UCB planning, since we take into account the variance values when identifying action sequences. A setting in which this assumption is obviously satisfied is when trajectories are allowed to start at arbitrary $z \in \mcl Z$, such as those points with maximal variance. 

Finally, we prove a useful Lemma that writes the iteration update in the predictive variance values in terms of the previous variance values.
\begin{lem} \label{lemma : variance update}
    Consider a GP utilizing kernel $k$ with predictive mean \eqref{eqn: predictive mean} and variance \eqref{eqn: predictive variance}. Upon observation at $z^\ast \in \mcl Z$, then the update to the predictive variance at the point $z \in \mcl Z$ can be written as
    \begin{equation}
        \sigma^2_{n+1}(z) = \sigma^2_{n}(z) -  \frac{\operatorname{cov}^2_{\sigma^2}\lp z^\ast, z\rp}{\sigma^2_n(z^\ast) + \sigma^2}.
    \end{equation}
\end{lem}
\begin{proof}
    First, we write the definition of the predictive variance upon observing at $z^\ast$:
    \begin{align}
        &\sigma^2_{n+1}(z) = k(z,z) \nonumber \\
        &- \lp \boldsymbol{k}_z \ k(z^\ast, z) \rp 
        \begin{pmatrix}
            G_n + \sigma^2 I & \boldsymbol{k}_{z^\ast} \\
            \boldsymbol{k}_{z^\ast}^T & k(z^\ast, z^\ast) + \sigma^2 \\
        \end{pmatrix}^{-1}
        \begin{pmatrix}
            \boldsymbol{k}_z \\
            k(z^\ast, z) \\
        \end{pmatrix} \\
        &=: k(z,z) - \lp \boldsymbol{k}_z \ k(z^\ast, z) \rp 
        W_{n+1}
        \begin{pmatrix}
            \boldsymbol{k}_z \\
            k(z^\ast, z) \\
        \end{pmatrix}.
    \end{align}
    Defining $W_n = (G_n + \sigma^2 I)^{-1}$ and $a^{-1} = k(z^\ast, z^\ast) + \sigma^2 - \boldsymbol{k}_{z^\ast}^T W_n \boldsymbol{k}_{z^\ast} = \sigma^2_n(z^\ast) + \sigma^2$, we can use the block matrix inversion formula to compute
    \begin{align}
        W_{n+1} &= \begin{pmatrix}
            W_n & 0 \\
            0 & 0 \\
        \end{pmatrix} + a \begin{pmatrix}
            W_n \boldsymbol{k}_{z^\ast} \boldsymbol{k}_{z^\ast}^T W_n & - W_n \boldsymbol{k}_{z^\ast} \\
            -\boldsymbol{k}_{z^\ast}^T W_n & 1 \\
        \end{pmatrix} \\
    &= \begin{pmatrix}
            W_n & 0 \\
            0 & 0 \\
        \end{pmatrix} + a 
        \begin{pmatrix}
            W_n \boldsymbol{k}_{z^\ast} \\ 
            -1 \\ 
        \end{pmatrix} 
        \lp \boldsymbol{k}_{z^\ast}^T W_n \ -1 \rp. 
    \end{align} 
Now, we can compute the inner product
    \begin{align}
        \lp \boldsymbol{k}_{z^\ast}^T W_n \ -1 \rp \begin{pmatrix}
            \boldsymbol{k}_z \\
            k(z^\ast, z) \\
        \end{pmatrix} &= -\lp k(z^\ast, z) - \boldsymbol{k}_{z^\ast}^T W_n\boldsymbol{k}_{z} \rp  \\
        &= - \operatorname{cov}_{\sigma^2}\lp z^\ast, z\rp,
    \end{align}
    which allows us to conclude
    \begin{align}
        \sigma^2_{n+1}(z) &= k(z,z) - \boldsymbol{k}_{z}^T W_n\boldsymbol{k}_{z} - a \lp \operatorname{cov}_{\sigma^2}\lp z^\ast, z\rp\rp^2 \\
        &= \sigma^2_n(z) - \frac{\operatorname{cov}^2_{\sigma^2}\lp z^\ast, z\rp}{\sigma^2_n(z^\ast) + \sigma^2},
    \end{align}
    which concludes the proof.
\end{proof}

With the stated assumptions and Lemma \ref{lemma : variance update}, we now state our main theorem regarding the convergence of variance values resulting from our proposed active sampling method.
\begin{thm} \label{thm : convergence of variance}
    If for all $\ep  \geq 1$ we have that Assumptions \ref{assumption : finite Z}-\ref{assumption : greater than average} are satisfied
    and that the regularization parameter is scaled as $\sigma^2 = \lp (\ep-1)N +  \iter\rp^{-2}$,
    then we have that sampling action sequences based on the finite-horizon planning using discrepancy-based GP-UCB acquisition function \eqref{eqn : mpc acquisition function setup} gives convergence of the posterior variance values, 
    \begin{equation}
        \max_{z \in \mcl Z} \ \sigma_{\ep, 0}^2(z) \rightarrow 0
    \end{equation}
    as $\ep \rightarrow \infty$.

    If Assumption \ref{assum: bounded rkhs norm} holds, then Lemma \ref{lemma : well-calibrated rkhs} implies consistency of the corresponding predictive mean $\mu_{\ep,0} \rightarrow f$ as $\ep \rightarrow \infty$.
\end{thm}
\begin{proof} 
    Let $\ep \geq 2$ denote a given episode and let $\iter_\ast \in \{1, \ldots, \Iter\}$ be the iteration index given by Assumption \ref{assumption : greater than average}. Furthermore, let $z_\ast \coloneqq z^{\ep, \iter_\ast}$ be the corresponding state-action pair selected at the identified iteration. Due to the monotonic decreasing nature of the variance over iterations and episodes, we can write 
    \begin{align} \label{eqn : max var z}
        \max_{z \in \mcl Z} \ \sigma^2_{\ep+1, 0}(z) &\leq  \sum_{z \in \mcl Z} \ \sigma^2_{\ep, \Iter}(z) \leq \sum_{z \in \mcl Z} \ \sigma^2_{\ep, \iter_\ast + 1}(z).
    \end{align}
    
    According to Lemma \ref{lemma : variance update}, the decrease in the variance for $z_\ast$ at the identified iteration $\iter_\ast$ can be written as
    \begin{align} \label{eqn : variance decrease}
        \sigma^2_{\ep, \iter_\ast + 1}(z_\ast) &= \sigma^2_{\ep, \iter_\ast}(z_\ast) - \frac{\sigma^4_{\ep, \iter_\ast}(z_\ast)}{\sigma^2 + \sigma^2_{\ep, \iter_\ast}(z_\ast)} \\
        &= \lp \frac{\sigma^2 }{\sigma^2 + \sigma^2_{\ep, \iter_\ast}(z_\ast)}\rp  \sigma^2_{\ep, \iter_\ast}(z_\ast) \\
        &\leq \lp \frac{((\ep-1)\Iter)^{-2}}{((\ep-1)\Iter)^{-2} + \sigma^2_{\ep, \iter_\ast}(z_\ast)}\rp  \sigma^2_{\ep, \iter_\ast}(z_\ast) \\
        &\leq \lp \frac{\Iter^{-2}}{\Iter^{-2} + \sigma^2_{\ep, \iter_\ast}(z_\ast)}\rp  \sigma^2_{\ep, \iter_\ast}(z_\ast),
    \end{align} 
    where we have used the property that $\sigma^2 = \lp (\ep-1)N +  \iter\rp^{-2} \leq ((\ep-1)\Iter)^{-2}$ and that $\ep \geq 2$.
    Plugging this into \eqref{eqn : max var z} and using \eqref{eqn : assumption bound} from Assumption \ref{assumption : greater than average}, we bound
   \begin{align}
        & \max_{z \in \mcl Z} \ \sigma^2_{\ep+1, 0}(z) \nonumber \\
        &\quad  \leq  \sum_{z  \not= z_\ast} \sigma^2_{\ep, \iter_\ast}(z) + \lp \frac{\Iter^{-2}}{\Iter^{-2} + \sigma^2_{\ep, \iter_\ast}(z_\ast)}\rp  \sigma^2_{\ep, \iter_\ast}(z_\ast)  \\
        &\quad \leq  \sum_{z  \in \mcl Z} \sigma^2_{\ep, \iter_\ast}(z) - \lp \frac{\sigma^2_{\ep, \iter_\ast}(z_\ast)}{\Iter^{-2} + \sigma^2_{\ep, \iter_\ast}(z_\ast)}\rp  \sigma^2_{\ep, \iter_\ast}(z_\ast)  \\
        &\quad \leq \lp 1 - \frac{1}{N_z}\lp \frac{\sigma^2_{\ep, \iter_\ast}(z_\ast)}{\Iter^{-2} + \sigma^2_{\ep, \iter_\ast}(z_\ast)} \rp \rp \sum_{z  \in \mcl Z} \sigma^2_{\ep, \iter_\ast}(z)  \\
        &\quad \leq \lp 1 - \frac{\sum_{z  \in \mcl Z} \sigma^2_{\ep, \iter_\ast}(z)}{N_z \lp N_z\Iter^{-2} + \sum_{z  \in \mcl Z} \sigma^2_{\ep, \iter_\ast}(z)\rp} \rp \sum_{z  \in \mcl Z} \sigma^2_{\ep, \iter_\ast}(z) \\
        &\quad \leq \lp 1 - \frac{\sum_{z  \in \mcl Z} \sigma^2_{\ep, 0}(z)}{N_z \lp N_z\Iter^{-2} + \sum_{z  \in \mcl Z} \sigma^2_{\ep, 0}(z)\rp} \rp \sum_{z  \in \mcl Z} \sigma^2_{\ep, 0}(z),
    \end{align}
    where in the last line we have used the monotonic decreasing property of the predictive variance values over the iterations.
    
    Defining $v_\ep = \frac{1}{N_z}\sum_{z \in \mcl Z} \sigma^2_{\ep, 0}(z)$, we see the recurrence relation 
    \begin{equation}
        v_{\ep + 1} - v_{\ep} \leq \frac{-v^2_{\ep}}{N_z(N_z\Iter^{-2} + v_{\ep})}
    \end{equation}
    with initial condition $v_1 = N_z$ (since $\sigma^2_{1, 0}(z) = 1$ for all $z \in \mcl Z$). 
    Noting that the function $f(v) = -v^2/(a(b + v))$ is decreasing on $v \in (0, \infty)$, we can pass to a corresponding differential equation to serve as an upper bound for $v_{\ep}$:
    \begin{align}
        v'(t) = \frac{-v^2}{N_z(N_z\Iter^{-2} + v)} \qquad v(0) = 1,
    \end{align}
    for which we can solve for the ``time'' $t(\epsilon) > 0$ for which $v(t(\epsilon)) = \epsilon$ by separation of variables:
    \begin{align}
        t(\epsilon) &= -N_z\lp N_z \Iter^{-2} \int_{N_z}^\epsilon \frac{1}{v^2}dv + \int_{N_z}^\epsilon \frac{1}{v}\, dv  \rp \\
        &= \frac{N_z(N_z-\epsilon)}{\Iter^{2}\epsilon} + N_z \log\lp\frac{N_z}{\epsilon}\rp. 
    \end{align}
    That is to say, by episode $\tau(\epsilon) \geq t(\epsilon)$, we are ensured that 
    \begin{equation}
        \max_{z \in \mcl Z} \sigma^2_{\ep(\epsilon),0}(z) \leq \epsilon. 
    \end{equation}
    With this explicit convergence rate, we can see that as $\ep \rightarrow \infty$ we then obtain that $\sigma_{\ep, 0}^2(z) \rightarrow 0$ for all $z \in \mcl Z$. 
\end{proof}

The proof of Theorem \ref{thm : convergence of variance} relies on the convenient fact that our discrepancy-based GP-UCB acquisition function reduces to maximum variance sampling in the first iteration of each episode. As such, the monotonicity of predictive variance in combination with the assumption that we can start trajectories at arbitrary $z \in \mcl Z$ gives that we can bound the variance at each iteration in an episode by the variance of the value $z_\ast$ given by Assumption \ref{assumption : greater than average}.

\section{Numerical Results} \label{sec: numerical results}

For all experiments, we use a Gaussian (RBF) kernel $k(x, x') = \exp(-\gamma \lVert x - x' \rVert^{2})$, $\gamma = 0.5$, and the regularization parameter is chosen to be $\sigma^2 = 1 / n^{2}$, where $n$ is the size of the dataset.

\subsection{Estimating a Simple Pendulum System}

\begin{figure}
    \centering
    \includegraphics[keepaspectratio]{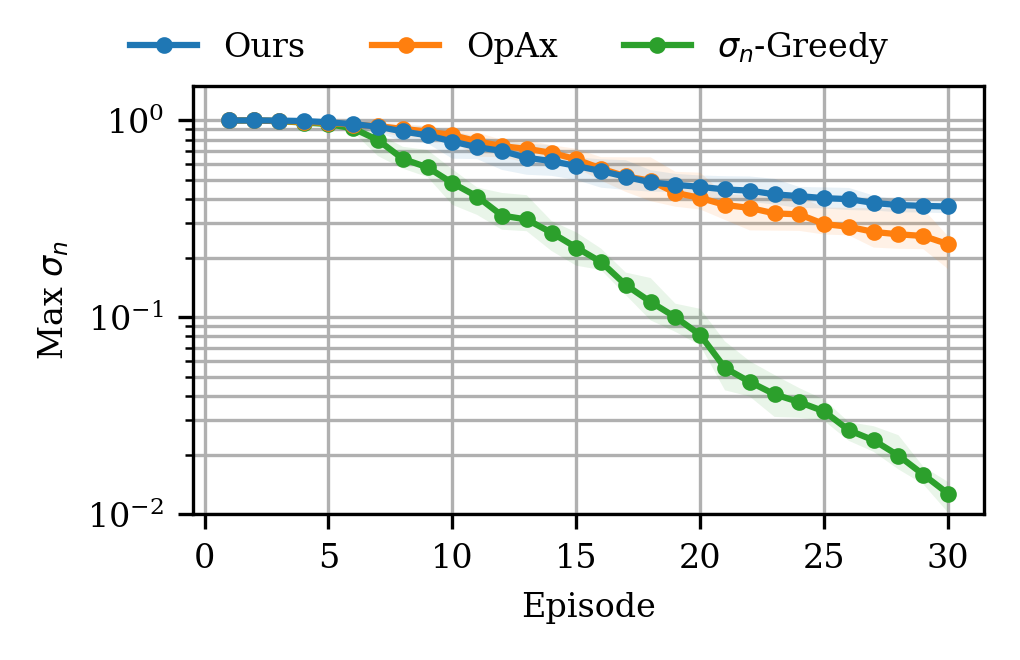}
    \caption{Average reduction of maximum variance (uncertainty) over $8000$ test points $\mathcal{T}$ spaced evenly over the entire state space. The shaded region shows the maximum and minimum values over $10$ independent trials.}
    \label{fig: max sigma comparison}
\end{figure}  

\begin{figure}
    \centering
    \includegraphics[keepaspectratio]{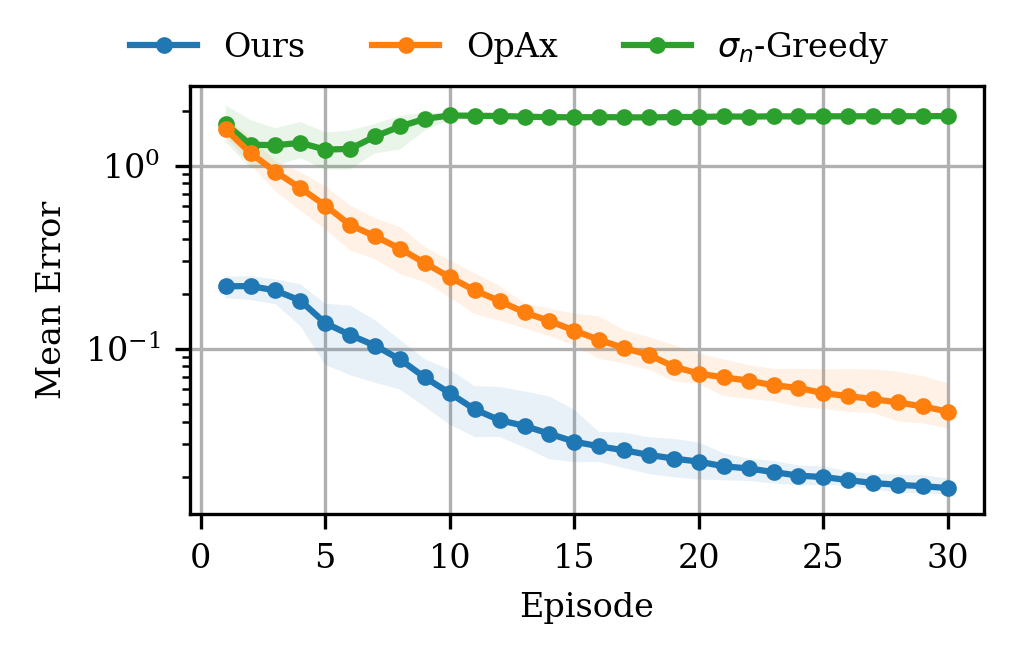}
    \caption{Mean squared error (MSE) of the learned models over test points $\mathcal{T}$. The inclusion of side information leads to reduced error. The shaded region shows the maximum and minimum values over $10$ independent trials.}
    \label{fig: mean error comparison}
\end{figure}

\begin{figure}
    \centering
    \includegraphics[keepaspectratio,width=\columnwidth]{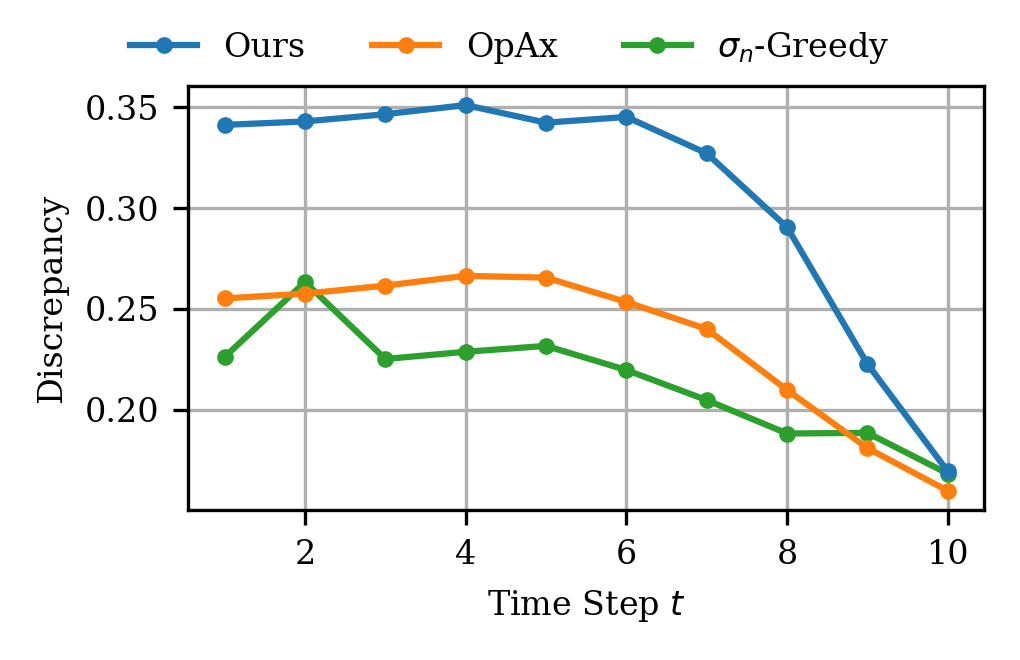}
    \caption{Average discrepancy between the prior model and the true, underlying dynamics at the points visited by the algorithms during each episode.}
    \label{fig: discrepancy}
\end{figure}

We first consider the problem of estimating the system dynamics of the controlled pendulum system as in \cite{brockman2016openai}. 
The equations of motion of the pendulum are given by $m l^{2} \ddot{\theta} + 3m g l \sin(\theta) = 3u$, where $g = 9.81$ is the acceleration due to gravity, the link mass is $m = 1$, and the link length is $l = 1$. The state of the system is given by the angle $\theta$ and angular velocity $\dot{\theta}$ of the pendulum, $x = [\theta, \dot{\theta}]^{\top} \in \mathbb{R}^{2}$, and the control input is the torque applied to the pendulum, $u \in \mathbb{R}$. For the benchmark, the angle $\theta$ is adjusted to be within the range $\theta \in [-\pi, \pi]$, the angular velocity is bounded such that $\dot{\theta} \in [-8, 8]$, and the control input is bounded to $u \in [-2, 2]$, meaning the system is under-actuated. 
We presume that the true dynamics are unknown but that we have access to an \emph{imperfect} prior model of the system with mis-specified parameters $g = 9.0$, link mass $m = 0.5$, and link length $l = 2.0$.

We implement two baselines for comparison. 
First, we consider OpAx \cite{sukhija2023optimistic}, which uses optimistic planning to identify action sequences that decrease predictive variance. At the beginning of each episode, OpAx computes an open-loop policy that maximizes the information gain during the episode using an optimistic estimate of the dynamics. It is called ``optimistic'' since it uses a planner that adds additional control variables to artificially steer plausible rollouts to states that maximize information gain \cite{NEURIPS2020_a36b598a}. 
Second, we consider a greedy variance-based exploration algorithm that myopically selects the control input at each iteration that has maximum variance. In other words, this $\sigma_{n}$-greedy policy is a purely exploration-based policy.

For all approaches, we use a GP model as in \eqref{eqn: gaussian process mean} and \eqref{eqn: gaussian process covariance} and sample over $30$ episodes with a time horizon of $\Iter=10$. For our approach and OpAx, we use the improved cross-entropy MPC planner (iCEM) \cite{pmlr-v155-pinneri21a} to compute the exploration policy. 
We use a planning horizon of $10$, with $10$ iterations, $50$ action sequence samples at each iteration, and an elite set size of $10$, holding $5$ elites between iterations. See \cite{pmlr-v155-pinneri21a} for more details.
In each episode, we choose the initial condition $x_{0}$ to be the point in $\mathcal{X}$ with the highest variance according to the current GP estimate of the dynamics \eqref{eqn : episodic gp model var}.

In order to compare our approach to OpAx and the $\sigma_{n}$-greedy policy, we compute the maximum uncertainty and mean prediction accuracy at a set of $8000$ test points $\mathcal{T} = \lbrace (x_{j}, u_{j}) \rbrace_{j=1}^{8000}$ spaced uniformly in the state space. 
Figure \ref{fig: max sigma comparison} shows the reduction in maximum variance and Figure \ref{fig: mean error comparison} 
shows the mean squared error (MSE) between the learned model $\mu_{\ep, \Iter}$ as in \eqref{eqn : episodic gp model} and the true underlying dynamics $f$. 
To demonstrate good performance, an active learning algorithm must balance exploration (reducing uncertainty or variance) while simultaneously improving accuracy. 

We see in Figure \ref{fig: max sigma comparison} that our approach performs comparably to OpAx at reducing uncertainty, though we explore the state-action space at a slightly reduced rate. Nevertheless, we see in Figure \ref{fig: mean error comparison} that our approach initially starts with a significantly reduced MSE due to the inclusion of side information via a non-zero mean prior. While it is possible to augment the GP model in OpAx to utilize such a prior, we note that the exploration procedure in OpAx does not take the prior into account during exploration.
This is key because it means that even though both methods explore at a similar rate, our approach focuses its exploration in areas with the highest mismatch (discrepancy) between the prior and the observations, providing higher intrinsic value for the learning task.
We can see this clearly in Figure \ref{fig: discrepancy}, which shows the average discrepancy between the prior model $p_0$ and the true dynamics $f$ of the visited states during each episode.
This is an important result, since it means we favor regions with high model mismatch while avoiding redundant regions where our prior model aligns closely with the true dynamics. 
As expected, we see in Figure \ref{fig: max sigma comparison} that the $\sigma_{n}$-greedy policy performs the best at reducing the uncertainty of the GP-based estimate, but note that this does not correspond to a commensurate improvement in the approximation quality. 

\subsection{Control Performance Using Learned Dynamics}

\begin{figure}
    \centering
    \includegraphics[keepaspectratio,width=0.49\columnwidth]{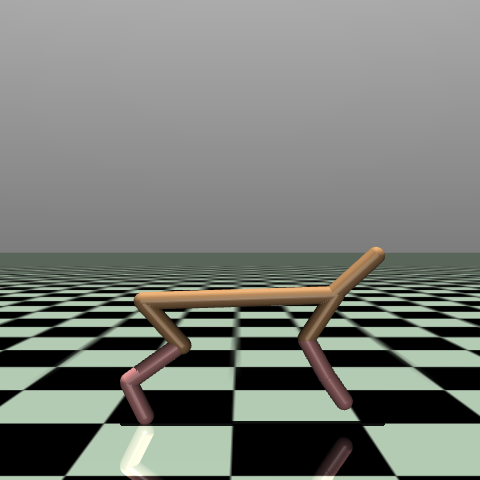} \hfill
    \includegraphics[keepaspectratio,width=0.49\columnwidth]{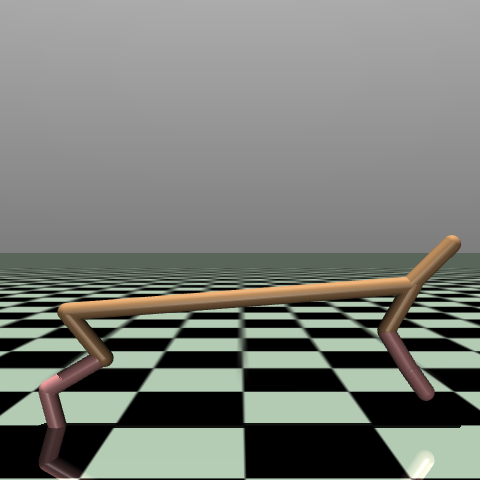}
    \caption{The true half-cheetah system (left) and the imperfect half-cheetah system used as the bias for our algorithm (right).}
    \label{fig: half-cheetah}
\end{figure}

\begin{figure}
    \centering
    \includegraphics{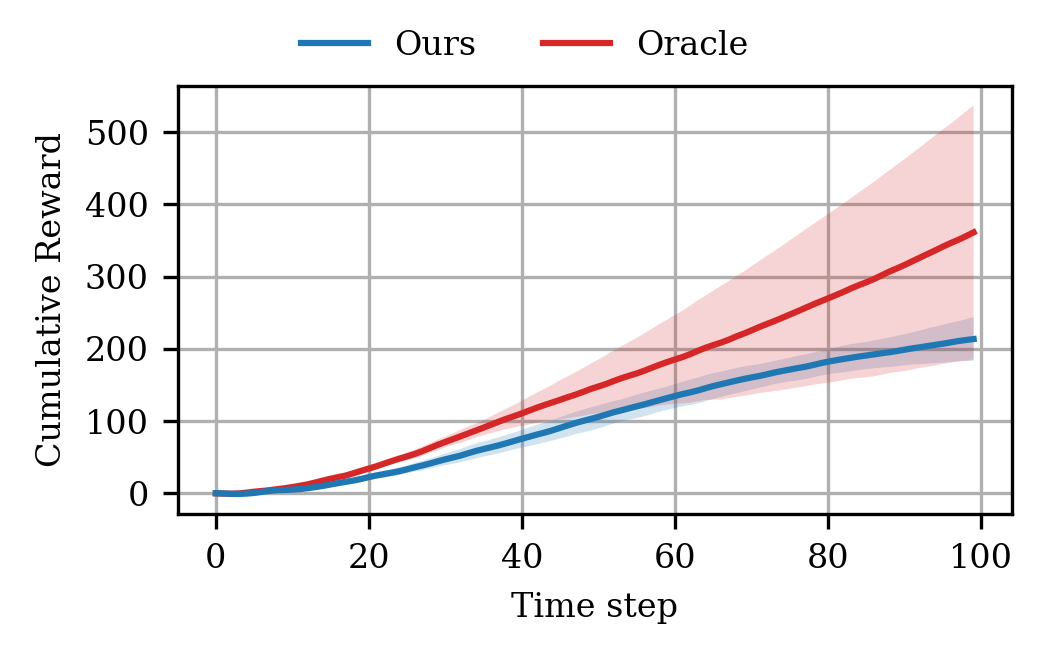}
    \caption{Cumulative reward of our approach at a downstream control task compared to the oracle using the true dynamics.}
    \label{fig: cumulative reward}
\end{figure}

We next consider the problem of computing a control policy using the model learned via our proposed active learning approach. Our goal is to measure whether the learned model is sufficiently high-fidelity for the purpose of control. 
We consider the half-cheetah environment from the MuJoCo benchmark suite \cite{6386109}. 
The half-cheetah system is specified by a collection of rigid links, joints, and actuators (see Figure \ref{fig: half-cheetah}). The state is in $\mathbb{R}^{18}$, consisting of the position and velocity of the various links, and the input space is in $\mathbb{R}^{6}$, representing the torques applied to each motor, bounded to be in $[-1, 1]$. 

We consider solving an optimal control problem as in \eqref{eqn: stochastic optimal control problem}, where the objective is to maximize the travel speed to the right (the positive $x$ direction), while minimizing control effort. We presume that the true system is unknown, but that we have access to side information in the form of an imperfect model with a torso that is $1.8$ times as long as the actual system (Figure \ref{fig: half-cheetah}, right). 

As before, we use the iCEM planner \cite{pmlr-v155-pinneri21a} to compute the exploration policy. We use a planning horizon of $10$, with $10$ iterations, $100$ action sequence samples at each iteration, and an elite set size of $10$, holding $5$ elites between iterations. 
Then, we fix the GP model of the dynamics and compute a separate policy to solve the optimal control problem as in \eqref{eqn: stochastic optimal control problem} using MPC, using the mean predictor $\mu_{\ep, \iter}$ as in \eqref{eqn : episodic gp model} for the predictive model of the dynamics in \eqref{eqn: stochastic optimal control problem dynamics}.
We similarly use iCEM to compute the control inputs, only using $200$ action sequence samples, an elite set size of $50$, and holding $15$ elites per iteration. 

Figure \ref{fig: cumulative reward} shows the average cumulative reward of the system over $10$ independent runs using the learned model as the predictive model. For comparison, we also computed the average cumulative reward using the actual system dynamics (oracle). We can see in Figure \ref{fig: cumulative reward} that our approach yields a model that demonstrates good empirical performance at a downstream control task.

\section{Conclusion}

In this paper, we present an active learning method that incorporates prior domain knowledge in the sampling procedure as well as the learned model. Under reasonable assumptions, we prove that our active sampling method provides a consistent estimator of the dynamics. Through numerical experiments, we demonstrate that our active learning approach produces an empirical dynamics estimate with lower error than methods that neglect prior knowledge, while simultaneously prioritizing exploration in regions that demonstrate a higher discrepancy between the prior model and our data-driven estimate.


\bibliographystyle{IEEEtran}
\bibliography{bibliography}

\end{document}